\documentclass[A4paper,11pt]{article}

\usepackage{latexsym}
\usepackage{epsfig}
\usepackage{amssymb}
\usepackage{amsmath}
\usepackage{psfrag}
\usepackage{paralist}

\newcommand{\ignore}[1]{}
\newtheorem{theorem}{Theorem}[section]

\newtheorem{lemma}[theorem]{Lemma}

\newtheorem{claim}[theorem]{Claim}

\newtheorem{definition}[theorem]{Definition}

\newtheorem{corollary}[theorem]{Corollary}

\newcommand{\cA}{{\cal A}}
\newcommand{\cB}{{\cal B}}
\newcommand{\cC}{{\cal C}}
\newcommand{\cD}{{\cal D}}

\newcommand{\cH}{{\cal H}}
\newcommand{\cU}{{\cal U}}
\newcommand{\cP}{{\cal P}}
\newcommand{\cQ}{{\cal Q}}
\newcommand{\cS}{{\cal S}}
\newcommand{\bP}{{\bf P}}
\newcommand{\bU}{{\bf U}}
\newcommand{\bD}{{\bf D}}
\newcommand{\bQ}{{\bf Q}}
\newcommand{\be}{{\bf e}}
\newcommand{\RR}{{\mathbb R}}

\newcommand{\df}{\textrm{def}}

\newcommand{\poly}{\textrm{poly}}

\newcommand{\qed}{\hfill $\Box$ \\}
\newcommand{\eps}{\varepsilon}

\begin{document}
\begin{titlepage}
\title{Is submodularity testable?}
\date{}

\author{
C. Seshadhri \\
IBM Almaden Research Center\\
650 Harry Road, \\
San Jose, CA 95120\\
{\tt csesha@us.ibm.com} \and
Jan Vondr\'ak \\
IBM Almaden Research Center\\
650 Harry Road, \\
San Jose, CA 95120\\
{\tt jvondrak@us.ibm.com}
}
\maketitle

\thispagestyle{empty}
\begin{abstract}

We initiate the study of property testing of submodularity on the boolean hypercube.
Submodular functions come up in a variety of applications in combinatorial
optimization. For a vast range of algorithms, the existence of an oracle
to a submodular function is assumed. But how does one check if this oracle
indeed represents a submodular function? 

Consider a function $f:\{0,1\}^n \rightarrow \RR$. The \emph{distance to
submodularity} is the minimum fraction of values of $f$ that need
to be modified to make $f$ submodular. If this distance is more
than $\epsilon > 0$, then we say that $f$ is $\epsilon$-far
from being submodular. The aim is to have
an efficient procedure that, given input $f$ that is $\epsilon$-far
from being submodular, certifies that $f$ is not submodular.
We analyze a very natural tester for this problem, and prove
that it runs in subexponential time. This gives the first non-trivial tester for submodularity. 
On the other hand, we prove an interesting lower bound (that is,
unfortunately, quite far from the upper bound) suggesting that
this tester cannot be very efficient in terms of $\epsilon$.
This involves non-trivial examples of functions which are far from submodular
and yet do not exhibit too many local violations.

We also provide some constructions indicating the difficulty in designing 
a tester for submodularity. We construct
a partial function defined on exponentially many points
that cannot be extended to a submodular function, but any strict
subset of these values can be extended to a submodular function. 
\end{abstract}
\end{titlepage}

\section{Introduction}

Submodular functions have been studied in great depth in combinatorial optimization
\cite{Edmonds70,NWF78,NWF78II,Lovasz83,Frank97,Schrijver00,FFI00}.
A set function $2^U \rightarrow \RR$ is submodular
if  $\forall S, T \subseteq U$, $ f(S \cup T) + f(S \cap T) \leq f(S) + f(T) $.
An alternative and equivalent view of submodularity is the monotonicity
of \emph{marginal values}. For all $S \subset T$ and elements $i \notin T$,
a submodular function satisfies $f(S \cup \{i\}) - f(S) \geq f(T \cup \{i\}) - f(T)$.
We will think of $f$ as a function in $\{0,1\}^n \rightarrow \mathbb{R}$.

These functions are often used in many algorithmic applications and very naturally
show up when modeling utilities. 
It is quite common to assume that algorithms have oracle access to some submodular
function: given a set $S$, we have access to $f(S)$. 
Observe that, in general, the description of the submodular function $f$ 
has size that is exponential in $n$, whereas
most algorithms that use $f$ run in polynomial time. This means that these algorithms
look at a very tiny fraction of $f$, yet their behavior depends on a very global
property of $f$. 
This leads to the very natural question: what if the function $f$
provided to the algorithm was \emph{not} submodular? Could the algorithm
detect this, or would it get fooled? Obviously, if $f$ is constructed by taking
a submodular function and making very few changes to the values, then
there is no need to think that algorithms should be affected. On the other
hand, if $f$ is ``significantly different" from a submodular function,
the behavior of these algorithms could very different. 

Let us formally explain the notion of being different from a submodular function.
Since polynomial time algorithms are \emph{sublinear} with respect to the size of $f$,
it is natural to use some property testing terminology.
A function $f$ is \emph{$\epsilon$-far from being submodular} if $f$ needs to be changed
at an $\epsilon$-fraction of values to make it submodular. 
In polynomial time, can we detect that such a function is not submodular?
If this is not possible, then this raises some very fundamental questions
about submodularity. If the plethora of algorithms used cannot tell whether
their input $f$ is submodular or not, then in what sense are they actually
using the submodularity of $f$? This would suggest that the algorithms
exploit a property more general than submodularity. It would
be strange if we expect input functions $f$ to have
a property (submodularity), but we cannot even check if these functions
deviate significantly from submodularity.

The main question here is whether submodularity is testable, i.e, is there
a polynomial time procedure that distinguishes submodular functions
from those that are $\epsilon$-far? (This question was first posed
as an open problem in~\cite{PRR03},
in the context of submodularity testing over grids.
Their results focused on testing over large low-dimensional grids rather
than the high-dimensional hypercube $\{0,1\}^n$.)
More concretely, what are the kind of structural properties
of submodularity that we need to address? Property testing algorithms, especially those
for functions on the hypercube, usually check for some \emph{local property}.
These algorithms check if the desired property holds in a small
local neighborhood, for some randomly chosen neighborhoods. If no
deviation is detected, then property testers conclude that the 
input function is close to the property. Do similar statements hold
for submodularity? We show non-trivial upper and lower bounds
connecting local submodularity violations to the distance.

Property testing proofs often show that a function is close
to a property by explicitly modifying the function to make
it have the property. Usually, there is some procedural
method to perform this conversion. This raises a very interesting
question about \emph{partial} submodular functions: suppose one is given
a \emph{partial} function over the hypercube. This means that 
some set of values is defined, but the remaining are left undefined.
Under what circumstances can this be completed into a submodular function?
If this cannot be completed, can we provide a small certificate
of this? For a vast majority of natural testable properties (over
functions on the hypercube, e.g. monotonicity)
such small certificates do exist. Unfortunately, this is no longer true
for submodularity. We present an example showing that a minimal certificate
of non-extendability can be exponentially large.

\subsection{Our results} \label{sec-res}


Before we state our main theorems, we first set some notation.

\begin{definition} \label{def:der}
Denote by $\be_i \in \{0,1\}^n$ the canonical basis vector which has $1$
in the $i$-th coordinate and $0$ everywhere else.

For a function $f:\{0,1\}^n \rightarrow \RR$, $i \in [n]$ and $x \in \{0,1\}^n$ such that
$x_i=0$, we define the marginal value of $i$ (or discrete derivative) at $x$ as
$ \partial_i f(x) = f(x + \be_i) - f(x).$

A function $f$ is submodular, if for any $i \in [n]$
and $x,y \in \{0,1\}^n$ such that $x_i = y_i = 0$ and $x \leq y$ coordinate-wise, 
$\partial_i f(x) \geq \partial_i f(y)$.

The distance $d(f,g)$ between two functions $f$ and $g$ is the fraction of points $x$
where $f(x) \neq g(x)$. Let $\cS$ be the set of all submodular functions.
The \emph{distance of $f$ to submodularity} is $\min_{g \in \cS} d(f,g)$. 
We say \emph{$f$ is $\epsilon$-far from being submodular} if the distance
of $f$ to submodularity is more than $\epsilon$.
\end{definition}

\begin{definition}
A \emph{property tester for submodularity} is an algorithm with the following
properties.
%
\begin{compactitem}
\item If $f$ is submodular, then the algorithm answers YES with probability $1$\footnote{We are 
actually dealing with one-sided testers here. If we allowed a probability of error for this case,
that would be a two-sided tester.}.
\item If $f$ is $\epsilon$-far from submodular, then the algorithm answers NO
with probability at least $2/3$.
\item The number of queries made to $f$ is \emph{sublinear} in the domain size,
which is $2^n$. (Ideally, the number of queries is polynomial in $n$ and $1/\epsilon$.)
\end{compactitem}
\end{definition}

\paragraph{Submodularity vs.~monotonicity.}
Our first observation is that testing submodularity is at least as hard as testing
monotonicity. More formally, the problem of testing monotonicity for a function
$f:\{0,1\}^n \rightarrow \RR$ can be reduced to the problem of testing
submodularity for a function $f':\{0,1\}^{n+1} \rightarrow \RR$. 
We present this reduction in Section~\ref{sec:reduction}.

A consequence of this is that known lower bounds for monotonicity testing apply also
to submodularity testing. For example, it is known that a non-adaptive monotonicity tester
requires at least $\Omega(\sqrt{n})$ queries \cite{FLNRR02}. We remark that the best
known monotonicity tester on $\{0,1\}^n$ takes $O(n^2/\eps)$ queries \cite{DGLRRS99} and
is non-adaptive.

Submodularity can be naturally viewed as ``second-degree monotonicity", i.e. monotonicity
of the discrete partial derivatives $\partial_i f$.
So a very natural test for submodularity is to simply run a monotonicity tester
on the functions $\partial_i f$. 
In one direction, it is clear that
for a submodular function, such a tester would always accept. However, it is not clear
whether this tester would recognize functions that are far from being submodular
and label them as such. 

Monotonicity testers search randomly for pairs $x, x+\be_i$
such that $f(x) > f(x+\be_i)$. Such a pair of points can be naturally called a ``violated pair".
It is known that if $f$ is $\epsilon$-far from being monotone, then the fraction of violated pairs
is at least $\epsilon/n^{O(1)}$ \cite{GGLRS00,DGLRRS99}.
If we want to test submodularity by reducing to a monotonicity tester in each direction, 
this means that we are looking for violations of the following type: $x \in \{0,1\}^n$
such that $x_i=x_j=0$ and $f(x+\be_i) - f(x) < f(x+\be_i+\be_j) - f(x+\be_j)$.
We call such violations {\em violated squares}.

\begin{definition} \label{def:sq}
We call $\{x,x+\be_i,x+\be_j,x+\be_i+\be_j\}$ a \emph{square}. This
is called a \emph{violated square}, if $ f(x) + f(x+\be_i+\be_j) > f(x+\be_i) + f(x+\be_j)$.
The {\em density} of violated squares is the number
of violated squares divided by ${n \choose 2} 2^{n-2}$.
\end{definition}

Our main combinatorial result consists of two bounds on the relationship of
the distance from submodularity and the density of violated squares.

\begin{theorem} \label{thm-tester}
Let $n$ be a sufficiently large integer.
\begin{itemize}
\item 
Let $\epsilon \in (0,e^{-5})$.
For any function $f:\{0,1\}^n \rightarrow \RR$ that is $\epsilon$-far from being submodular,
the density of violated squares is at least $\epsilon^{O(\sqrt{n} \log n)}$.
\item 
For any $\epsilon \geq 2^{-n/10}$,
there is a function $f:\{0,1\}^n \rightarrow \RR$
which is $\epsilon$-far from being submodular and its density of violated squares
is less than $\epsilon^{4.8}$.
\end{itemize}
\end{theorem}

The first part of the theorem is proven through relatively basic observations.
The second part is quite technical and requires a much deeper understanding of submodularity.

Theorem~\ref{thm-tester} provides evidence that testing submodularity
is very different from testing monotonicity.
An intuition one might get from monotonicity testing is that if a natural extension
to submodularity exists, its dependence on $\epsilon$ should be relatively mild,
perhaps linear or quadratic. We show that this is not the case, in particular
if the dependence is a polynomial in $1/\epsilon$,
the degree of the polynomial would have to be at least $5$. 
This holds even in the range of exponentially small $\epsilon = 2^{-\Theta(n)}$,
which means that $poly(n) / \epsilon^{4.8}$ queries for any polynomial in $n$ are
not enough. This might be interpreted
as counterintuitive to the notion that the dependence is polynomial at all. However,
we cannot currently push this construction any further.

The first part of Theorem~\ref{thm-tester} implies immediately that a submodularity tester that checks $q =
 1/\epsilon^{O(\sqrt{n} \log n)}$ random squares succeeds 
with high probability\footnote{We use ``high probability" to refer to probability $> 2/3$.}.
Note that this is a non-adaptive tester, because the queries do not depend
on the function values.
To our knowledge, this is the first testing result asymptotically better than
the trivial tester checking $2^{\Theta(n)}$ squares.

\begin{corollary} \label{cor-tester}
There is a subexponential time non-adaptive tester for submodularity.
This procedure samples $1/\epsilon^{O(\sqrt{n} \log n)}$
sqaures at random and checks if any are violated. If the input $f$ is $\epsilon$-far from being
submodular, this procedure rejects with high probability.
\end{corollary}

\paragraph{Extending partial functions.}
A \emph{partial function} $f$ is one that is defined on only some
subset of the hypercube. Such a function is \emph{extendable},
if the remaining values can be filled in to get a submodular function.
Although the question of extending partial functions is interesting
in itself, it also has some relevance to question of testing
submodularity. 

Any proof of a property tester must show that if a function $f$ passes
the tester (with high probability), then $f$ must be $\epsilon$-close
to submodularity. This is usually done by arguing that if $f$ has a 
sufficiently low density of local violations,
one can modify an $\epsilon$-fraction of values and remove
all ``obstructions" to submodularity. Since an $f$ that passes
the tester must have a low density of local violations, 
$f$ is $\epsilon$-close.
An understanding of these obstructions to submodularity
is often helpful for designing testers.
An obstruction is
just a subset of values that cannot exist in any submodular function.

Given a partial function $f$ that is not extendable, we would ideally 
like to find a small certificate for this property. 
Unfortunately, we will show that such certificates can be exponentially large. 
We give a partial function with a surprising property.
The partial function $f$ is defined on an exponentially
large set and is not extendable. If any \emph{single}
value is removed, then this new function is extendable.

\begin{definition} \label{def-min} For a partial function $f$,
let $\df(f)$ be the set of domain points when $f$ is defined. Let $\cA \subseteq \{0,1\}^n$. 
The \emph{restriction of $f$ to $\cA$}, $f|_{\cA}$,
is the partial function that agrees with $f$ on $\cA$ and is undefined everywhere
else. The partial function $f$ is \emph{minimally non-extendable} if 
$f|_\cA$ is extendable for all $\cA \subset \df(f)$.
\end{definition}

\begin{theorem} \label{thm-min} There exists a minimally non-extendable function $f$
such that $|\df(f)| = 2^{\Omega(n)}$.
\end{theorem}

%
%
%
%
%
%

\subsection{The difficulty in testing submodularity} \label{sec-diff}

The values of $f$ can interact in non-trivial ways to create
obstructions to submodularity. Contrast this to monotonicity.
A partial function $f$ (on the hypercube) cannot be extended to a non-decreasing monotone function
iff there is a pair of sets $S \subset T$ such that $f(S) > f(T)$. 
There is always a certificate of size $2$ that a partial function
cannot be extended. So
this completely characterizes the obstructions to monotonicity, and is
indeed one of the reasons why monotonicity testers work. 
Our work implies that such a simple characterization does not exist for submodularity.
Indeed, as Theorem~\ref{thm-min} claims, obstructions to submodularity
can have an extremely complicated structure.

Functions that are far from being submodular can ``hide" their bad behavior.
In Theorem~\ref{thm:1-vsq}, we show the existence of a function $f$ with
exactly \emph{one} violated square, but making $f$ submodular requires
changing $2^{n/2}$ values. Somehow, even though the function is (in a weak sense)
``far" from submodular, the only local violation that manifests itself is a single
square. The functions described by the second part of Theorem~\ref{thm-tester} 
are constructed through generalizations of this example.

%
%

\subsection{Previous work} \label{sec-prev}

Property testing, which was defined in~\cite{RS96,GGR98},
is a well-studied field of theoretical computer science.
Efficient testers have been given for a wide variety of combinatorial,
algebraic, and geometric problems (see surveys ~\cite{FischerSurvey, GoldreichSurvey, RonSurvey}).
The problem of property testing for monotonicity over the hypercube has been
studied in~\cite{GGLRS00, DGLRRS99, FLNRR02, E04, FR, BCGM10}. 
In particular, monotonicity of a function over $\{0,1\}^n$
can be tested using $O(n^2 / \epsilon)$ non-adaptive queries \cite{DGLRRS99}
and $\Omega(\sqrt{n})$ queries are necessary \cite{FLNRR02}.

As mentioned earlier, the problem of testing submodularity was first raised
first by~\cite{PRR03}. They considered submodularity over general
grid structures (of which the hypercube is a special case).
Their focus was on testing submodularity over $2$-dimensional
grids.
Specifically,~\cite{PRR03} gave strong results for testing \emph{Monge matrices}.
Monge matrices are essentially submodular functions over the $n \times m$ integer grid. 
Here, the dimension is $2$, but the domain in each component is large.
In contrast, we are studying submodular functions over high-dimensional domains,
where each component is binary.
Hence, our problem is quite orthogonal to testing Mongeness, and we
need a different set of techniques.

Another related set of results is recent work on learning and approximating
submodular functions \cite{GHIM09,BH09}. Here, we want to examine a value oracle
through polynomially many queries (which is similar to our setting)
and learn sufficient information so that we are able to answer queries
about the function. The difference is that in this model,
we care about {\em multiplicative-factor} approximation to the original function.
An even more essential difference is that the input function is guaranteed to be submodular,
rather than possibly being corrupted. 
For example, \cite{GHIM09} shows that we can ``learn" a monotone submodular function
using polynomially many queries so that afterwards we can answer value queries
within a multiplicative $\widetilde{O}(\sqrt{n})$ factor, and this is optimal up to
logarithmic factors. In contrast, the input function in our model might
be masquerading as a submodular function but in truth be very far from being submodular.

\subsection{Organization}
The rest of the paper is organized as follows.
In Section~\ref{sec:tester}, we present our basic submodularity tester
and prove the first part of Theorem~\ref{thm-tester}.
In Section~\ref{sec:lower-bounds}, we present our construction of submodular functions
from lattices and prove the second part of Theorem~\ref{thm-tester}.
In Section~\ref{sec:LP}, we discuss extendability of submodular functions
and prove Theorem~\ref{thm-min}.
In Section~\ref{sec:reduction}, we present the reduction from monotonicity testing
to submodularity testing.
In Section~\ref{sec:future}, we discuss future directions.

\section{A subexponential submodularity tester}
\label{sec:tester}


\paragraph{The violated-square tester.}
\begin{compactitem}
\item For a parameter $q \in {\mathbb Z}$, repeate the following $q$ times.
\item Sample uniformly at random $x \in \{0,1\}^n$ and $i,j \in \{\ell: x_\ell = 0\}$.
If $$f(x) + f(x+\be_i+\be_j) > f(x+\be_i) + f(x+\be_j),$$ i.e. if $\{x,x+\be_i,x+\be_j,x+\be_i+\be_j\}$
is a violated square, then return NO.
\item If none of the tested squares is violated, then return YES.
\end{compactitem}

\medskip

Clearly, if the input function is submodular, the tester answers YES. We would like to understand
how well this tester performs in case the input function is $\epsilon$-far from being submodular.
The following observation is standard and reduces this question to a combinatorial problem
about violated squares.

\begin{lemma}
\label{lem:tester-eq}
The following two statements are equivalent:
\begin{compactitem}
\item The violated-square tester using $q(n,\epsilon)$  queries detects every function that is $\epsilon$-far from submodular
with constant probability.
\item For every function which is $\epsilon$-far from submodular, the density of violated squares is
$\Omega(1 / q(n,\epsilon))$.
\end{compactitem}
\end{lemma}

Therefore, to understand this tester we need to understand the relationship between the distance from submodularity
and the density of violated squares. In the rest of this section, our main goal is to prove
the first part of Theorem~\ref{thm-tester},
i.e. the claim that for a function $\epsilon$-far from submodular, the density of violated squares must be at least
$\epsilon^{O(\sqrt{n} \log n)}$. Using Lemma~\ref{lem:tester-eq}, this implies Corollary~\ref{cor-tester}.
First, we prove the following lemma. 

\begin{lemma}
\label{lem:easy-fix}
Assume $\{x,x+\be_i,x+\be_j,x+\be_i+\be_j\}$ is a violated square. Then it is possible to decrease
all the values either in $\{y: y \leq x\}$ or in $\{y: y \geq x+\be_i+\be_j\}$ by a constant
such that the square $\{x,x+\be_i,x+\be_j,x+\be_i+\be_j\}$ is no longer violated and no new violated square is created.
\end{lemma}

\begin{proof}
Denote by $d = f(x) + f(x+\be_i+\be_j) - f(x+\be_i) - f(x+\be_j)$ the ``deficit" of the violated square.
One way to fix this square is to decrease the value of $f(x)$ by $d$; however, this might create new violated squares.
Instead, we decrease the value of $f(y)$ for every $y \leq x$; i.e., we define a new function $\tilde{f}(y) = f(y) - d$
for $y \leq x$, and $\tilde{f}(y) = f(y)$ otherwise. (Alternatively, we can define $\tilde{f}(y) = f(y) - d$ for
$y \geq x+\be_i+\be_j$, and $\tilde{f}(y) = f(y)$ otherwise; the analysis is symmetric and we omit this case.)

Consider any other square that was previously not violated, i.e.~$f(x') + f(x'+\be_{i'}+\be_{j'})
 \leq f(x'+\be_{i'}) + f(x'+\be_{j'})$. Note that $x'_{i'} = x'_{j'} = 0$.
We consider four cases:
\begin{itemize}
\item If $x'_\ell > x_\ell$ for some coordinate $\ell$, then we do not modify any value in the square
$\{x',x'+\be_{i'},x'+\be_{j'},x'+\be_{i'}+\be_{j'}\}$.
\item If $x' \leq x$ and both $x_{i'} = 0$ and $x_{j'} = 0$, then the only value we modify
in the square is $f(x')$, which is decreased by $d$. This cannot create a submodularity violation.
\item If $x' \leq x$ and exactly one of the coordinates $x_{i'}, x_{j'}$ is $1$, then we modify two values in the square;
for example $f(x')$ and $f(x'+\be_{i'})$. Since we decrease both by the same amount, this again cannot create
a submodularity violation.
\item If $x' \leq x$ and $x_{i'} = x_{j'} = 1$, then we decrease all four values in the square by the same amount.
Again, this cannot create a submodularity violation. \qed
\end{itemize}
\end{proof}

This means we can fix violated squares one by one, and the number of violated squares decreases by one
every time. The cost we pay for each fix is the number of points in the cube above or below the respective square.
Recall that we count the number of modified values overall, and hence what counts is the union of all the cubes
modified in the process. Intuitively, it is more frugal to choose up-closed cubes for violated squares that are above the middle
layer of the hypercube, and down-closed cubes for squares that are below the middle. A counting
argument gives the following.

\begin{lemma}
\label{lem:easy-fix-dist}
Let $\epsilon \in (0,e^{-5})$ and let $f$ have at most $\epsilon^{\sqrt{n} \log n} 2^n$ violated squares.
Then these violated squares can be fixed by modifying at most $\epsilon 2^n$ values.
\end{lemma}

\begin{proof}
Denote by $B$ the set of bottom points for the violated squares which are below the middle layer; i.e.
we have $||x||_1 \leq n/2$ for each $x \in B$. (The squares above the middle layer can be handled symmetrically.)
We choose to modify the down-closed cube, $C_x = \{ y \in \{0,1\}^n: y \leq x\}$, for each $x \in B$.
We can fix the violated square one by one, by modifying values in the cubes $C_x$. The total number of modified
values is $| \bigcup_{x \in B} C_x|$. We estimate the cardinality of this union by combining two simple bounds across
levels of the hypercube. Denote $L_j = \{x \in \{0,1\}^n: ||x||_1 = j\}$. We have
$$ \Big| \bigcup_{x \in B} C_x\Big| = \sum_{j=0}^{n/2} \Big| \bigcup_{x \in B} (C_x \cap L_j) \Big|.$$
First, by the union bound, we have
$$ \Big|\bigcup_{x \in B} (C_x \cap L_j)\Big| \leq \sum_{x \in B} |C_x \cap L_j| = \sum_{x \in B} {||x||_1 \choose j}
 \leq |B| {n/2 \choose j}. $$
Secondly, we have (trivially)
$$ \Big|\bigcup_{x \in B} (C_x \cap L_j)\Big| \leq |L_j| = {n \choose j}.$$
We choose the better of the two bounds depending on $j$. In particular, for $j \leq n/2 - a\sqrt{n}$, we get
$ \sum_{j=0}^{n/2-a\sqrt{n}} {n \choose j}  = 2^n \Pr[X \leq n/2 - a\sqrt{n}] \leq 2^n e^{-a^2} $
where $X$ is a binomial $Bi(n,1/2)$ random variable and the last inequality is a standard Chernoff bound.
For $j > n/2 - a \sqrt{n}$, we use
$ \sum_{j=n/2-a\sqrt{n}}^{k} |B| {k \choose j} = |B| \sum_{j=0}^{a\sqrt{n}} {k \choose j} \leq |B| k^{a \sqrt{n}} \leq |B| n^{a \sqrt{n}}. $
We conclude that
$$ \Big| \bigcup_{x \in B} C_x\Big| = \sum_{j=0}^{n/2} \Big| \bigcup_{x \in B} (C_x \cap L_j) \Big|
 \leq 2^n e^{-a^2} + |B| n^{a \sqrt{n}}.$$
Let $a = \frac12 \ln (1/\epsilon)$; we also assume that $|B| \leq 2^n \epsilon^{\sqrt{n} \ln n}$.
For $\epsilon \in (0,e^{-5})$, this implies 
$$ \Big| \bigcup_{x \in B} C_x\Big| \leq 2^n e^{-(\frac12 \ln (1/\epsilon))^2} + 2^n \epsilon^{\sqrt{n} \ln n} n^{\frac12 \sqrt{n} \ln (1/\epsilon)}
= (\epsilon^{\frac14 \ln (1/\epsilon)} + \epsilon^{\frac12 \sqrt{n} \ln n}) 2^n \leq \frac12 \epsilon 2^n. $$
\qed
\end{proof}

This lemma immediately implies the first part of Theorem~\ref{thm-tester}.
Assuming that $f$ is $\epsilon$-far from being submodular,
we get that the number of violated squares is at least $\epsilon^{\sqrt{n} \log n} 2^n$ for $\epsilon \in (0,e^{-5})$,
i.e. the density of violated squares is at least $\epsilon^{\sqrt{n} \log n}$.

\section{Few violated squares, yet large distance}
\label{sec:lower-bounds}

We now give a construction of submodular functions that have
large distance but a relatively small fraction of violated squares.
As we mentioned earlier, these bounds are nowhere near our positive results.
Nonetheless, we are able to show a significant difference from monotonicity. 

Our first tool to construct these functions is an interesting
family of submodular functions. 
It is known that that the set of minimizers of a submodular function always
forms a lattice\footnote{A lattice is any partial order with the operations of "meet" and "join".
In our setting, this means a subset of $\{0,1\}^n$ closed under
taking coordinate-wise minimum and maximum. Or equivalently, a family
of sets closed under taking intersections and unions.} \cite{Edmonds70}.
We prove that conversely, for any lattice ${\cal L} \subset \{0,1\}^n$ there is a submodular
function whose set of minimizers is exactly $\cal L$.
We will then piece together these submodular functions to construct
a non-submodular function with the desired properties.

\subsection{Submodular functions from lattices}

\begin{lemma}
\label{lem:lattice}
Let ${\cal L} \subset \{0,1\}^n$ be a lattice, i.e a set of points closed under coordinate-wise
minimum and maximum. Then the following Hamming distance function is submodular:
$$ d_{\cal L}(x) = \min_{y \in {\cal L}} ||x-y||_1. $$
\end{lemma}

\begin{proof}[Lemma~\ref{lem:lattice}]
In this proof, we use the set-function notation and identify $\{0,1\}^n$ with subsets of $[n]$.
A lattice ${\cal L} \subset \{0,1\}^n$ is a family of sets closed under taking unions
and intersections. The distance function $d$ can be written as
$$ d(S) = \min_{L \in {\cal L}} |S \Delta L| $$
where $|S \Delta L|$ denotes the symmetric difference. 
Assume that $d(S) = |S \Delta U|$ and $d(T) = |T \Delta V|$ for some $U,V \in {\cal L}$.
We want to prove $d(S \cup T) + d(S \cap T) \leq d(S) + d(T)$. We prove in fact that
$$ |(S \cup T) \Delta (U \cup V)| + |(S \cap T) \Delta (U \cap V)| \leq |S \Delta U| + |T \Delta V| $$
which is sufficient since $U \cup V, U \cap V \in {\cal L}$ by the lattice property,
and therefore $d(S \cup T) \leq |(S \cup T) \Delta (U \cup V)|, d(S \cap T) \leq |(S \cap T) \Delta (U \cap V)|$.
These two symmetric differences can be bounded as follows:
\begin{eqnarray*}
|(S \cup T) \Delta (U \cup V)| & = &
 |(S \cup T) \setminus (U \cup V)| + |(U \cup V) \setminus (S \cup T)| \\
 & = & |S \cap \bar{U} \cap \bar{V}| + |\bar{S} \cap T \cap \bar{U} \cap \bar{V}|
  + |U \cap \bar{S} \cap \bar{T}| + |\bar{U} \cap V \cap \bar{S} \cap \bar{T}| \\
  & \leq & |S \cap \bar{U} \cap \bar{V}| + |\bar{S} \cap T \cap \bar{V}|
  + |U \cap \bar{S} \cap \bar{T}| + |\bar{U} \cap V \cap \bar{T}|,
\end{eqnarray*} 
\begin{eqnarray*}
|(S \cap T) \Delta (U \cap V)| & = &
 |(S \cap T) \setminus (U \cap V)| + |(U \cap V) \setminus (S \cap T)| \\
 & = & |S \cap T \cap \bar{V}| + |S \cap T \cap \bar{U} \cap V|
  + |U \cap V \cap \bar{T}| + |U \cap V \cap \bar{S} \cap T| \\
 & \leq & |S \cap T \cap \bar{V}| + |S \cap \bar{U} \cap V|
  + |U \cap V \cap \bar{T}| + |U \cap \bar{S} \cap T|.
\end{eqnarray*} 
Adding up the two bounds and merging terms such as
$|S \cap \bar{U} \cap \bar{V}| + |S \cap \bar{U} \cap V| = |S \cap \bar{U}|$,
we obtain
\begin{eqnarray*}
|(S \cup T) \Delta (U \cup V)| + |(S \cap T) \Delta (U \cap V)| & \leq &
 |S \cap \bar{U}| + |T \cap \bar{V}| + |U \cap \bar{S}| + |V \cap \bar{T}|  
 =  |S \Delta U| + |T \Delta V|.
\end{eqnarray*}
\qed
\end{proof}

Considering the known fact that the minimizers of any submodular function
form a lattice, we get the following characterization.

\begin{corollary}
Let ${\cal S} \subseteq \{0,1\}^N$.
Then the following statements are equivalent:
\begin{compactenum}
\item $\cal S$ is a lattice.
\item $\cal S$ is the set of minimizers of some submodular function.
\item The Hamming distance function $d_{\cal S}(x) = \min_{y \in {\cal S}} ||x-y||_1$
is submodular.
\end{compactenum}
\end{corollary}

\subsection{Functions with one violated square}

We start with the following counter-intuitive result.

\begin{theorem}
\label{thm:1-vsq}
For any $n$, there is a function $f:\{0,1\}^n \rightarrow \RR$ which has
exactly one violated square but $2^{n/2}$ values must be modified to make it submodular.
\end{theorem}

We remark that this statement is tight in the sense that for any function with
exactly one violated square, it is sufficient to modify $2^{n/2}$ values 
(we leave the proof as an exercise, using Lemma~\ref{lem:easy-fix}).
To prove Theorem~\ref{thm:1-vsq}, we use Lemma~\ref{lem:lattice} which says that any
lattice in $\{0,1\}^n$ yields a natural submodular function. This function does not have
any violated squares. However, we will add two additional dimensions and extend
the function in such a way that each point of the lattice will produce exactly
one violated square. Moreover, due to the nature of the distance function,
the function we construct will be a linear function in a large neighborhood
of each violated square. This will imply that we cannot simply change one value
in each violated square if we want to make the function submodular - such changes
would propagate and force many other values to be changed as well.
We make this argument precise later. The construction is as follows.

\paragraph{Construction.}
{\em Given:} Lattice ${\cal L} \subset \{0,1\}^n$.
{\em Output:} Function $f:\{0,1\}^{n+2} \rightarrow \RR$.
\begin{compactitem}
\item We denote the arguments of $f$ by $(a,b,x)$ where $x \in \{0,1\}^n$ and $a,b \in \{0,1\}$.
\item Let $f(0,0,x) = ||x||_1 = \sum_{i=1}^{n} x_i$.
\item Let $f(1,1,x) = 1 - ||x||_1 = 1 - \sum_{i=1}^{n} x_i$.
\item Let $f(0,1,x) = f(1,0,x) = d_{\cal L}(x)$, the Hamming distance function from $\cal L$.
\end{compactitem}

\begin{lemma}
\label{lem:L-vsq}
The function $f(a,b,x)$ constructed above has exactly $|{\cal L}|$ violated squares, of the form
$\{ (0,0,x)$, $(0,1,x)$, $(1,0,x)$, $(1,1,x) \}$ for each $x \in {\cal L}$.
\end{lemma}

\begin{proof}
Observe that for any fixed $a,b \in \{0,1\}$, $f(a,b,x)$ is a submodular function of $x$.
Therefore, there is no violated square $\{ z, z+\be_i, z+\be_j, z+\be_i+\be_j \}$
unless at least one of $i,j$ is a special bit. 

If exactly one of $i,j$ is a special bit, we can assume that it is the first special bit.
First assume the other special bit is $0$, therefore we are looking at a  square with values
$f(0,0,x), f(1,0,x), f(0,0,x+\be_i, f(1,0,x+\be_i)$. By construction,
we know that $f(0,0,x+\be_i) - f(0,0,x) = 1$ and $f(1,0,x+\be_i) - f(1,0,x)
= d_{\cal L}(x+\be_i) - d_{\cal L}(x) \leq 1$, therefore the square cannot be violated.
Similarly, if the other special bit is $1$, we are looking at a square with values
$f(0,1,x), f(1,1,x), f(0,1,x+\be_i, f(1,1,x+\be_i)$. Here,
we always have $f(1,1,x+\be_i) - f(1,1,x) = -1$, and $f(0,1,x+\be_i) - f(0,1,x)
 = d_{\cal L}(x+\be_i) - d_{\cal L}(x) \geq -1$. So again, the square cannot
be violated.

Finally, consider a square where $i,j$ are exactly the special bits. The square has values
$f(0,0,x)$, $f(0,1,x)$, $f(1,0,x)$, $f(1,1,x)$. Observe that $f(0,0,x) + f(1,1,x) = 1$,
and $f(0,1,x) + f(1,0,x) = 2 d_{\cal L}(x)$. The square is violated if and only if
$2 d_{\cal L}(x) < 1$, i.e. when $x \in {\cal L}$. This means that we have a one-to-one
correspondence between violated squares and the points of the lattice.
\qed
\end{proof}

Thus we can generate functions with a prescribed number of violated squares, depending
on our initial lattice $\cal L$. The simplest example is generated by ${\cal L} = \{x\}$
being a 1-point lattice. In this case, it is easy to verify directly that the function
$d_{\cal L}(x)$ is submodular, and hence our construction produces exactly one violated
square.

The second part of our argument, however, should be that such a function
is not very close to submodular. In particular, consider ${\cal L} = \{x\}$ where $||x||_1=n/2$.
Suppose that we want to modify some values so that the function $f$ becomes submodular.
We certainly have to modify at least one value in the violated square
$\{ (a,b,x): a,b \in \{0,1\}\}$. However, for each fixed choice of $a,b \in \{0,1\}$,
the function $f(a,b,x)$ is linear. The last point in our argument is that it is impossible
to modify a small number of values ``in the middle" of a linear function (with many values
both above and below), so that the resulting function is submodular. First, we prove the following.

\begin{lemma}
\label{lem:increase-value}
Suppose $f:\{0,1\}^n \rightarrow \RR$ is a submodular function and $f(0) > 0$.
Then there are at least $2^{n-1}$ points $x \in \{0,1\}^n$ such that $f(x) \neq 0$.
\end{lemma}

Note that this is tight, for example by taking $f(x) = 1 - x_1$.
\medskip

\begin{proof}
We prove the statement by induction on $n$. Obviously it is true for $n=1$.
For $n>1$, we partition the cube $\{0,1\}^n$ as follows: let
$$ Q_i = \{ x \in \{0,1\}^n: x_1=\ldots=x_{i-1}=0, x_i=1 \}.$$
In other words, $Q_i$ is the set of points such that the first nonzero coordinate is $x_i$.
We have $\{0,1\}^n = \{0\} \cup \bigcup_{i=1}^{n} Q_i$. Now consider a submodular function
$f:\{0,1\}^n \rightarrow \RR$ such that $f(0) > 0$. We consider two cases.

If there is coordinate $i$ such that $f(\be_i) \leq 0$, then the discrete derivative
$\partial_i f(0)$ is negative. By submodularity, $\partial_i f$ must be negative everywhere.
Hence, for any point $x$ such that $x_i=0$, at least one of $f(x), f(x+\be_i)$ is nonzero.

The other case is that $f(\be_i) > 0$ for all $i \in [n]$. Then we apply the inductive
hypothesis to $Q_i$, which implies that at least $\frac12 |Q_i|$ values in $Q_i$ are nonzero.
By adding up the contributions from $Q_1,\ldots,Q_n$, we conclude that at least half
of all the values in $\{0,1\}^n$ are nonzero. 
\qed
\end{proof}

To rephrase the lemma, we can start with a zero function on $\{0,1\}^n$, increase
the value of $f(0)$ to a positive value, and ask - how many other values do we have to modify
to make the function submodular? The lemma says that at least $2^{n-1}$ values must be modified.
In fact, the condition of submodularity does not change under the addition of a linear function,
so the zero function can be replaced by any linear function.
Thus the lemma says that it is impossible to increase the value of a linear
function at the lowest point of a cube, without changing a lot of other values in the cube.

Note that it is possible to {\em decrease} the value of a linear function at the lowest point
of a cube and this does not create any violation of submodularity. What is impossible
is to decrease the value ``in the middle" of a linear function, without changing
a lot of other values. This is the content of the next lemma.

\begin{lemma}
\label{lem:decrease-value}
Suppose $n$ is even, $f:\{0,1\}^n \rightarrow \RR$ is a submodular function and $f(x) < 0$
for some $||x||_1 = n/2$. Then there are at least $2^{n/2}$ points $x \in \{0,1\}^n$
such that $f(x) \neq 0$.
\end{lemma}

This lemma is also tight, by taking $f(y) = -1$ whenever $y \leq x$ and $f(y) = 0$ otherwise.
\medskip

\begin{proof}
Consider $Q = \{ y \in \{0,1\}^n: y \leq x \}$; this is a cube of dimension $n/2$,
hence $|Q| = 2^{n/2}$. If $f(y) \neq 0$ for all $y \in Q$, we are done. Therefore,
assume that there is any point $y \in Q$ such that $f(y) = 0$.
Then consider a monotone path from $y$ to $x$; there must be an edge $(y',y'+\be_i)$
of negative marginal value. By submodularity, all edges $(z',z'+\be_i)$ for $z' \geq y'$
must have negative marginal value. There are at least $2^{n/2}$ such edges, since
all the $n/2$ zero bits in $x$ are also zero in $y'$ and can be increased arbitrarily
to obtain a point $z' \geq y'$. Each of these (disjoint) edges $(z',z'+\be_i)$
contains a point of nonzero value, and hence there are at least $2^{n/2}$ such points.
\qed
\end{proof}

Now we can complete the proof of Theorem~\ref{thm:1-vsq}.
\medskip

\begin{proof}[Theorem~\ref{thm:1-vsq}]
Consider the function $f: \{0,1\}^{n+2} \rightarrow \RR$ defined for a 1-point lattice
${\cal L} = \{x\}$, $||x||_1=n/2$. By Lemma~\ref{lem:L-vsq}, $f$ has exactly one
violated square.
Note that for each fixed $a,b \in \{0,1\}$, the function
$f(a,b,x)$ is linear as a function of $x$.

Suppose $f': \{0,1\}^{n+2} \rightarrow \RR$ is submodular (presumably close to $f$).
Since $f$ has a violated square $\{(0,0,x), (0,1,x), (1,0,x), (1,1,x) \}$,
$f'$ must differ from $f$ on at least one of these values. Fix $a,b \in \{0,1\}$ such that
$f'(a,b,x) \neq f(a,b,x)$ and consider 
the function $f'(a,b,x)-f(a,b,x)$ as a function of $x$.
Since $f$ is linear, $f'-f$ is again submodular as a function of $x$.
We have $(f'-f)(x) \neq 0$. If $(f'-f)(x) > 0$, we apply Lemma~\ref{lem:increase-value}
to the cube $\{y: y \geq x\}$; if $(f'-f)(x) < 0$, we apply Lemma~\ref{lem:decrease-value}.
In both cases, we conclude that there are at least $2^{n/2}$ values $x \in \{0,1\}^n$
such that $f'(x) \neq f(x)$. Therefore, $f$ is $2^{-n/2}$-far from submodular.
\qed
\end{proof}

\subsection{Boosting the example to increase distance}

Observe that in Theorem~\ref{thm:1-vsq}, the relationship between relative distance
and density of violated squares is quadratic: we have relative distance $\epsilon = 2^{-n/2}$
and density of violated squares $\simeq \epsilon^2 = 2^{-n}$.
In order to prove the second part of Theorem~\ref{thm-tester}, we need to consider a denser lattice.
Since the regions of linearity will be more complicated here, we need a more general
statement to argue about the number of values that must be fixed to make a function
submodular.

\begin{lemma}
\label{lem:downmon-increase}
Let $f: \{0,1\}^n \rightarrow \RR$ be submodular (non-increasing marginals)
on a down-monotone subset ${\cal D} \subset \{0,1\}^n$.
If $f(0) > 0$ then there are at least $\frac{1}{n+1} |{\cal D}|$ points $y \in {\cal D}$
such that $f(y) \neq 0$.
\end{lemma}

This is also tight - consider for example ${\cal D} = \{0, \be_1, \ldots, \be_n \}$
and $f(x) = 1 - ||x||_1$.
\medskip

\begin{proof}
Suppose $f(y) = 0$ for some $y \in {\cal D}$. Then let $x \leq y$ be minimal such that
$f(x) \leq 0$. Since $x$ is minimal (and cannot be $0$ because $f(0)>0$), for any $x_i=1$
we have $f(x-\be_i) > 0$. Hence $f(x) - f(x-\be_i) < 0$ and by submodularity
$f(y) - f(y-\be_i) < 0$. Since $f(y) = 0$, this implies that $f(y-\be_i)>0$.
In this case we call $y-\be_i$ a witness for $y$.

To summarize, for every $y \in {\cal D}$ we have either $f(y) \neq 0$ or $f(y - \be_i) \neq 0$
for some witness of $y$. Since every point can serve as a witness for at most $n$ other points,
the number of nonzero values must be at least $|{\cal D}| / (n+1)$.
\qed
\end{proof}

Now we are ready to prove the second part of Theorem~\ref{thm-tester}.

\medskip

\begin{proof}
We define ${\cal L} \subset \{0,1\}^n$ as follows:
\begin{compactitem}
\item Consider $n$ even and partition $[n]$ into pairs $\{2i-1,2i\}, 1 \leq i \leq n/2$.
\item Let ${\cal L} = \{ x \in \{0,1\}^n: \forall i;  x_{2i-1} = x_{2i} \}.$
\end{compactitem}
Obviously, this is a lattice, in fact it is isomorphic to a cube of dimension $n/2$.
The function $f: \{0,1\}^{n+2} \rightarrow \RR$ based on this lattice
has exactly $2^{n/2}$ violated squares, due to Lemma~\ref{lem:L-vsq}.
It remains to estimate the distance of $f$ from being submodular.

To that end, focus on the ``middle layer" of the lattice, ${\cal M} = \{x \in {\cal L}: ||x||_1=n/2\}$.
Such points have exactly a half of the pairs equal to $(0,0)$ and a half equal to $(1,1)$.
For each such point $x$, consider points $y \geq x$ such that $y$ still has the same number
of pairs equal to $(1,1)$ as $x$. Formally, let
$$ Q_x = \{ y \geq x: \forall i; y_{2i-1}=y_{2i}=1 \Rightarrow x_{2i-1}=x_{2i}=1 \}.$$
The reason for this definition is that for any point $y \in Q_x$, it is possible to trace it back to $x$
(by zeroing out all the pairs which are not equal to $(1,1)$, we obtain $x$). Hence the sets $Q_x$
are disjoint. The path from $y$ to $x$ is also the shortest possible path to any point of the lattice
(because it is necessary to modify all pair which are equal to $(1,0)$ or $(0,1)$). 
In other words, $d_{\cal L}(y) = ||x-y||$ for any $y \in Q_x$.
This implies that the function $f(a,b,y)$ for any fixed $a,b$ is linear as a function of
$y \in Q_x$.

Our final argument is that in order to make $f$ submodular, we would have to fix many values
in each set $Q_x$. Let us assume that $f'$ is submodular.
Since $f$ has a violated square $\{(0,0,x)$, $(0,1,x)$, $(1,0,x)$, $(1,1,x) \}$ for each
$x \in {\cal L}$,  $f'$ must be different from $f$ in at least one point in each such square.
More specifically, $f'$ must be larger than $f$ for one of the points $(0,1,x), (1,0,x)$
or $f'$ must be smaller than $f$ for one of the points $(0,0,x), (1,1,x)$.

Fix $a,b$ so that $f'(a,b,x)$ differs from $f(a,b,x)$ as above. Since $f$ is linear on $Q_x$,
we have $f'-f$ submodular on $Q_x$ and $(f'-f)(a,b,x) \neq 0$.
If $a \neq b$, we must have $(f'-f)(a,b,x) > 0$. Then applying
Lemma~\ref{lem:downmon-increase} to the set $Q_x - x$, we conclude that $f'-f$ must be
nonzero on at least $\frac{1}{n}|Q_x|$ points in $Q_x$. 

In the other case,
$a=b$, we have $(f'-f)(a,b,x) < 0$. Note that in this case $f$ is actually linear on all of $\{0,1\}^n$
and $f'-f$ is submodular everywhere. Then we use arguments similar to Lemma~\ref{lem:decrease-value}.
Let $Q_x^-$ be the set of points $y \leq x$ such that the set of $(0,0)$ pairs is the same
in $y$ and $x$. Again, $y \in Q_x^-$ can be traced back to $x$ and so these sets are disjoint.
From the proof of Lemma~\ref{lem:decrease-value}, we obtain that either $f(y) \neq 0$
for all $y \in Q_x^-$, or else there is an edge $(x-\be_i, x)$ of negative marginal value.
This implies that all edges above this edge have negative marginal value.
I.e., at least half of the points in $Q_x \cup (Q_x - \be_i)$ must have nonzero value.

Now let us count the size of $Q_x$. We have $n/4$ pairs of value $(0,0)$ which can be modified
and we have $3$ choices for each (we avoid $(1,1)$ for such pairs). Therefore,
$|Q_x| = 3^{n/4}$. The same holds for $Q_x^-$.

This holds for every lattice point in the middle layer $\cal M$. 
Therefore, each lattice point $x \in M$ contributes $\Omega(3^{n/4} / n)$ nonzero points
in $f'-f$. There are ${n/2 \choose n/4} = \Omega(2^{n/2} / n)$ points in $\cal M$.
We have to be careful about the last case where the nonzero points
are guaranteed to be in $Q_x \cup (Q_x - \be_i)$ rather than $Q_x$. Such points
could be potentially overcounted $n$ times, but we had a $1/2$-fraction of them
nonzero, so we still get $\Omega(3^{n/4} / n)$ nonzero points from each point in $\cal M$.
Overall, we get $\Omega(2^{n/2} 3^{n/4})$ nonzero points in $f'-f$. This means
that the distance of $f$ from being submodular is $\epsilon = \Omega(2^{-n/2} 3^{n/4})$.
A calculation reveals that this is $\epsilon \simeq \Omega(2^{-0.104 n})$, while the density
of violated squares is $2^{-n/2}  < \epsilon^{4.8}$.

Finally, it is easy to boost this example to larger value of $\epsilon$.
Supppose we want to construct an example for a given $n$ and
$\epsilon = 2^{-0.104 n'}$, $n'<n$ ($n'$ can even be a constant).
Assume for simplicity that $n = an'$ and $a$ is an integer. 
Then we start from an example on $n'$ coordinates where the distance is
$\epsilon = 2^{-0.104 n'}$ and density of violated squares is $2^{-n'/2}$.
We extend $f$ to dimension $n' = an$ so that it does not depend on the new coordinates.
There are no violated squares involving the new coordinates 
and hence the density of violated squares as well as relative distance remain unchanged.
\qed
\end{proof}

\section{Path certificates for submodular extension}
\label{sec:LP}

Given a partial function $f$, can we get a precise characterization of when $f$
is submodular-extendable? Using LP duality, we can give a combinatorial
condition that captures this condition. In this subsection, $f$ will be
some fixed partial function. We will set $\cD = \df(f)$ and $\cU = \cB \setminus \cD$.
Let us associate a variable
$x_S$ for every set $S$. If $S \in \cD$, then $x_S$ has value $f(S)$ (so this
is not really a variable, but it will be convenient to keep this notation).
For set $S$, $A^+(S)$ is the set $\{e = (S,S+i) | \ i \notin S\}$
and $A^-(S)$ is the set $\{e = (S-i,S) | \ i \in S\}$.
For edge $e=(S,S+i)$,
$\Gamma^+(e)$ is the set $\{e' = (S+j,S+i+j) | \ j \notin S\}$.
The set $\Gamma^-(e)$ is $\{e' = (S-j,S+i-j) | \ j \in S-i\}$.
If $f$ is extendable, then the following LP has a feasible solution.
\begin{eqnarray*}
\forall e, e' \in \Gamma^+(e), & x_e - x_{e'} \geq 0\\
\forall e = (S,S+i), & x_e - x_{S+i} + x_S \geq 0\\
& {\bf x} \geq {\bf 0}
\end{eqnarray*}
Using Farkas' lemma, if this is infeasible, then we can derive
a contradiction from these equations. So, we have dual variables
$y_{e,e'}, y_{e}$ associated with each equation, and the following
LP is feasible.
\begin{eqnarray*}
\forall e, & y_e + \sum_{e' \in \Gamma^+(e)} y_{e,e'} = \sum_{e' \in \Gamma^-(e)} y_{e',e}\\
\forall S \in \cU, & \sum_{e \in A^+(S)} y_e = \sum_{e \in A^-(S)} y_e \\
\forall e, e' \in \Gamma^+(e), & y_{e,e'} \geq 0 \\
& \sum_{S \in \cD} [\sum_{e \in A^-(S)} y_e - \sum_{e \in A^+(S)} y_e]f(S) < 0
\end{eqnarray*}
\begin{definition} \label{def-cert} Consider a set of directed paths {\bf P} 
consisting of cycles or paths with
endpoints in $\cD$. An edge is \emph{upward} if it is directed from the smaller
set to the larger, and \emph{downward} otherwise. 

Let $\bU$ be the multiset of upward edges
of $\bf P$ and $\bD$ be the multiset of downward edges (so we keep as many copies
of edge $e$ as occurrences in $\bf P$). Let $G$ be a bipartite graph 
on $\bU$ and $\bD$ (with links, instead of edges).
An edge $e \in \bU$ is linked
to $e' \in \bD$ if $e \preceq e'$. The set of paths $\bf P$ is \emph{matched} if
there is a perfect matching in $G$. 

The \emph{value} of a directed path $\cP$, $val(\cP)$, that
starts at $S \in \cD$ and ends at $S' \in \cD$ is $f(S') - f(S)$. Cycles have
value $0$. The value
of $\bf P$ is the sum of values of the paths in $\bf P$.
If $\bf P$ has negative value, then $\bP$ is referred to as a \emph{path
certificate}.
\end{definition}

\begin{lemma} \label{lem-matched} The partial function $f$ is not submodular-extendable iff
$f$ contains a path certificate.
\end{lemma}

\begin{proof} Suppose $\bf P$ is a path certificate, but
$f$ can be extended to a submodular function $f'$. Let $\bU$ be the multiset of upward edges
in $\bf P$ and $\bD$ the multiset of downward edges. We have a perfect matching between $\bU$
and $\bD$. Consider a matched pair $(e,e')$. We have $e \preceq e'$. By the submodularity
of $f'$, $f'(e) \geq f'(e')$. Considering $e,e'$ as directed edges, we get $f(e) + f'(e) \geq 0$.
Summing over all matched pairs,
$\sum_{e \in {\bf P}}f'(e) \geq 0$. Consider a path ${\cal P} \in {\bf P}$.
Note that $val(\cP)$ is the same in $f$ and $f'$, since $f'$ extends $f$.
Considering $\cP$ as a multiset of directed edges, we have $val(\cP) = \sum_{e \in \cP} f'(e)$.
We get $\sum_{\cP \in \bP} val(\cP) \geq 0$. Contradiction.

Suppose $f$ cannot be extended to a submodular function. By Farkas' lemma,
the second LP is feasible. Consider the directed hypercube (abusing notation, call this graph $\cB$). The second
equality is a flow conservation constraint for all vertices in $\cU$. Hence,
we can think of the $y_e$'s as 
giving a flow in $\cB$,
where the terminals are $\cD$. 
Precisely, $y_e$ is the flow in $e$ from the lower end to the higher end.
The first constraint is a little 
stranger\footnote{By that we mean, somewhat different, and not an unknown
dwarf.}. Consider the graph $G$, where the vertices are \emph{edges}
of the hypercube, and there is a directed link from $e$
to every member of $\Gamma^+(e)$. This actually
gives $n$ disconnected graphs, each of which is a hypercube
in $n-1$ dimensions. Think of $y_{e,e'}$ as a flow 
in $G$. Note that this is always positive. We do not really
have a flow conservation condition, because of the extra $y_e$.
Add a extra terminal for every $e$ that is attached 
to the vertex $e \in G$. This is called the terminal $e \in G$.
Think of $y_e$ amount of flow being removed (if $y_e \geq 0$)
or injected (if $y_e < 0$) into $e$ from this terminal. Then, we have a legitimate
flow in $G$ represented by the $y_{e,e'}$'s.

Since the $y$ values are rational, we can assume that they are integral.
We will construct a path certificate through a flow decomposition process.
At an intermediate stage, we will maintain a set $\bP$ of directed paths in $\cB$ and
a list of matched pairs in $\bP$. 
For each matched pair, we have a directed path in $G$ from the smaller
edge to the larger (call this
set of paths $\bQ$). All these paths start and end at terminals
in their respective graphs.
We maintain the following invariants. Through every path in $\bP \cup \bQ$,
a single unit of flow can be simultaneously routed, in the
flow given by the $y$ values.
Furthermore, a directed edge $e$ in $\bP$ is upward iff $y_e > 0$.
Flow in any directed edge of $\bQ$ is always positive.
Suppose the current set of paths $\bP$ is not completely
matched.
We will describe a procedure that either increases the number
of matched pairs, or adds a new path to both $\bP$ and $\bQ$.
That means that the total flow that is routed through $\bP$ (and $\bQ$)
increases by one. Since the flow is finite, this process must terminate
and return a set of matched paths.

Suppose there is an unmatched edge $e \in \bP$ (wlog, we can
take it to be upward).
This means that $y_e$ is positive. Note that because $\bP$
can be considered as a multiset of edges, there could
be many copies of the upward edge $e$ in $\bP$. Suppose
there are $t$ copies, which means that $t$ paths in $\bP$
pass through $e$. Since we can route one unit of flow
in each of these path simultaneously, $y_e \geq t$.
Let us look at the situation in $G$. At most $t-1$ copies
of $e$ are matched, so there are at most $t-1$ paths
in $Q$ that end at the terminal $e \in G$ (since $y_e \geq 0$,
there is a net influx at terminal $e \in G$). Let us route a single unit
of flow through all paths in $\bQ$ (and remove
this flow). This must still leave at one unit of flow
going into $e$. So, we can route one unit of flow
from some $e'$ to $e$ along path $Q$. Note that because the flow
is always positive in $G$, $e' \succ e$. 

Note that $y_{e'} < 0$, because in $G$, the terminal $e'$
has a net outflow. Suppose there is an unmatched copy
of $e'$ in $\bP$ (it must be downward). Then we can match $e$ to this copy
of $e'$, and we are done. Suppose this is not the
case. Let $s$ be the number of copies of the downward
edge $e'$ in $\bP$ (all of these are matched). We argue that $s < |y_{e'}|$. 
Suppose, for the sake of contradiction,
that $|y_{e'}| = s$. Them, there are $s$ paths in $\bQ$ that
start at the terminal $e' \in G$. If we remove
all the flow paths corresponding to $\bQ$, then there
is no flow going out of $e'$. But, we were able to route
one unit of flow from $e'$ to $e$ along $Q$ \emph{after}
removing flow corresponding to $\bQ$. Contradiction.
Hence $|y_{e'}| > s$. This means that after removing
all the flow corresponding to $\bP$ (in $\cB$), there
is still at least one unit of (downward) flow left on $e'$.
So, after the removal, we can still route one unit of flow
through $e'$, giving us path or cycle $P$. We add $P$ to $\bP$
and $Q$ to $\bQ$, observing that the invariants are
maintained. This ends the procedure.

Finally, we end up with a set of matched paths $\bP$. If this
has negative value, we have found our certificate. Suppose
it has positive value. We argue that the we can find a new 
(integral) solution for the dual which has a smaller flow.
This is done by just removing one unit flow along
all paths in the final $\bP$ and $\bQ$. Consider some upward
edge in $\bP$. Since $\bP$ is completely matched,
the number of copies of $e$ in $\bP$ is exactly
the number of paths in $\bQ$ ending at terminal $e$ in $G$.
Hence, the $y$ values, after the decrease, will
maintain the flow conservation conditions. The original
value of the solution is negative, and we removed
a set of matched paths of positive value. So, the value
of the remaining solution is still negative. This gives
us the new solution for the dual.
\qed
\end{proof}

A path in $\bP$ is called a \emph{singleton} if it consists of only a single edge.
We will prove some ``clean-up" claims that provide us with nice path certificates.

\begin{claim} \label{clm-cancel} Let $f$ be a partial function.
Let $f$ contain a set of matched paths $\bP$
and let $e$ is an upward edge in $\bP$ that is matched
to a downward copy of itself.
There is an operation that converts $\bP$ to $\bP'$ such that $\bP'$
contains the same multiset of edges $\bP$ except for an upward and downward
copy of $e$. The matching of $\bP'$ is identical to $\bP$ (except for the
matched pair of $e$) and $val(\bP) = val(\bP')$.
\end{claim}

\begin{proof} Let $e = (S,S+i)$. Suppose path
$\cP_u$ contains edge $e$ upwards, and $\cP_d$ contains it downwards. We can split $\cP_u$
into portions $\cP_{1,u}$ and $\cP_{2,u}$ such that the former is the part before $e$
and the latter is after $e$. Similarly, we can get $\cP_{1,d}$ and $\cP_{2,d}$. Note
that $\cP_{1,u}$ ends at $S$ and $\cP_{2,d}$ starts at $S$. Similarly, $\cP_{2,u}$ ends
at $S+i$ and $\cP_{1,d}$ starts at $S$. We can combine $\cP_{1,u}$ and $\cP_{2,d}$
to get a path $\cP'_1$. Similarly, we get $\cP'_2$. We replace $\cP_u$ and $\cP_d$
by he $\cP'_1$ and $\cP'_2$. Note that the sum of values
does not change. Also, the only edges removed are the upward and downward copies of $e$
and the matching on the remaining edges stays the same.
\qed
\end{proof}

\begin{claim} \label{clm-single} Let $f$ be partial function such that
for any square of $\cB$, at most $2$ points are present in $\df(f)$.
Let $f$
contain a path certificate $\bP$, such that no edge
occurs both upward and downward in $\bP$.
There exists a path certificate $\bQ$ 
such that $\bQ$ contains no singleton edge. Furthermore, no edge in $\bQ$ appears
both upward and downward.
\end{claim}

\begin{proof} We will show how to remove any singleton in $\bP$ and give
an ``equivalent" certificate $\bQ$. The value will remain the same.
Suppose there is a singleton path consisting of upward edge $e$.
Some downward edge $e'$, $e' \succeq e$ must occur in path $\cP \in \bP$. 
If $e = e'$, then this edge occurs both upward and downward. This cannot happen.
So $e' \succ e$. Let $e = (S,S+i)$ and $e' = (T+i,T)$, for some $S \subset T$.
We will split $\cP$ into two paths. Let $\cP_1$ be the portion of $\cP$ before $e'$
and $\cP_2$ be the portion after $e$. Note that $\cP_1$ ends at $T+i$ and $\cP_2$
starts at $T$. Consider a downward path $\cQ_1$ from $T+i$ to $S+i$ 
and a parallel upward path $\cQ_2$ from $S$ to $T$. Observe that there is a
perfect matching between the edges of $\cQ_1$ to those of $\cQ_2$. 

Consider the path
$\cQ'_1$ formed by joining $\cP_1$ to $Q_1$, and the similarly constructed $\cQ'_2$.
Note that $\cQ'_1$ ends at $S+i$ and $\cQ'_2$ starts at $S$. 
To get $\bQ$, we remove the singleton $e$
from $\bP$ and replace $\cP$ by $\cP_1$ and $\cP_2$. The set $\bQ$ is completely
matched. The edges in $\cQ_1$ and $\cQ_2$ (matched to each other) are disjoint.
Hence, no edge in $\bQ$ appears both upward and downward.
The singleton edge $e$
starts at $S$ and ends at $S+i$. So $val(\cQ'_1) + val(\cQ'_2) = val(e) + val(\cP)$.
and $val(\bQ) = val(\bP)$.
Suppose $|\cQ_1| > 1$. Then neither of $\cQ'_1$ and $\cQ'_2$ are singletons.
Suppose $\cQ_1$ is a single edge. Then $e$ and $e'$ form a square, so
neither endpoint of $e$ can be in $\df(f)$. This means that the path
$\cP_1$ and $\cP_2$ are at least of length $1$ and $\cQ'_1$ and $\cQ'_2$
are at least of length $2$. The total number
of singletons has decreased by $1$.
We can repeatedly apply this procedure, and remove
all singletons. 
\qed
\end{proof}

\subsection{Large minimal certificates}

This will require many steps. 
We will start by giving a construction of a long cycle in $\cB$ 
with some special properties. This cycle will be a sort of ``frame"
on which we can define $f$. For this $f$, we will find a set
of matched path of negative value, showing that $f$ is non-extendable.

The simple cycle will be obtained by performing a series of moves
in $\cB$. An \emph{upward} (resp. \emph{downward}) step is one where some coordinates
is incremented (resp. \emph{decremented}).
We will assume that $n = 2m+4$. 
The cycle will only involve points in the $m+1, m+2, m+3, m+4$
levels of $B$. We will call these levels the $1,2,3,4$ levels.
Any point is represented
as $(b_1,b_2,b_3,b_4,S,T)$, where $b_i$'s are bits, and $S$ and $T$
are sets on $m$ elements. We will denote the starting (and hence, ending) point
of the cycle to be $(0,0,1,0,\emptyset,[m])$, where $[m]$ represents the complete
set on $m$ elements.
The cycle $\cC$ has the following properties:
\begin{itemize}
	\item The cycle is simple, i.e., does not intersect itself.
	\item The cycle can be divided into a sequence of contiguous \emph{chunks} of three steps.
Every odd (resp. even) chunk has three upward (resp. downard) steps. There are an
even number of chunks.
	\item The cycle has $M \geq 2^{m}$ chunks.
	\item Let the $i$th chunk is denoted by $K_i$. The second edge $e$ of $K_{i}$
is parallel to the first edge $e'$ of $K_{i+1(\textrm{mod} \ M)}$. Suppose $i$ is odd.
Then $K_{i}$ has upward steps, and hence $e' \succ e$. Similarly,
if $i$ is even, $e' \prec e$.
\end{itemize}

A crucial combinatorial property of the hypercube that we use is the
existence of Hamiltonian circuits. We set $\cH$ to be a (directed) Hamiltonian
circuit on the $m$-dimensional hypercube. For any set $R \in \cH$,
$s(R)$ denotes the successor of $R$ in $\cH$. The \emph{complement path}
$\overline{\cH}$ is the Hamiltonian circuit obtained by taking
the set-complement of every point in $\cH$.

\begin{lemma} \label{lem-cycle} There exists a cycle $\cC$ with the properties
above.
\end{lemma}

\begin{proof} Starting from a point $(0,0,1,0,R,\overline{R})$,
we will give a sequence of $4$ chunks that will end at $(0,0,1,0,s(s(R)),\overline{s(s(R))})$.
Since $\cH$ is a Hamiltonian circuit, we get a cycle.
The reason we keep $R$ and $\overline{R}$ is that from $(\cdots,R,\overline{R})$,
we can perform a single upward and then downward step
to reach $(\cdots,s(R),\overline{s(R)})$.
We will assume that the moves to both $s(R)$ and $s(s(R))$
are upward.
Whenever this is not the case, we can
just reverse the roles of $R$ (or $s(R)$) and $\overline{R}$ (or $\overline{s(R)}$).

We describe the sequence of chunks. In the arrows below, the labels above them
represents the coordinate being changed. The numbers $1,2,3,4$ represent
the first four coordinates. If the label has a set, then that set
is being changed by moving along (appropriately) either $\cH$ or $\overline{\cH}$.
These labels help verify the matching property. The first
and third chunks only have upward steps, and the remaining
have only downward steps. For convenience, $S = s(R)$ and $T = s(S)$.

\begin{enumerate}
	\item $(0,0,1,0,R,\overline{R}) \stackrel{1}{\rightarrow} (1,0,1,0,R,\overline{R}) \stackrel{2}{\rightarrow} (1,1,1,0,R,\overline{R}) \stackrel{R}{\rightarrow} (1,1,1,0,S,\overline{R})$.
	\item $(1,1,1,0,S,\overline{R}) \stackrel{2}{\rightarrow} (1,0,1,0,S,\overline{R}) \stackrel{3}{\rightarrow} (1,0,0,0,S,\overline{R}) \stackrel{\overline{R}}{\rightarrow} (1,0,0,0,S,\overline{S})$.
	\item $(1,0,0,0,S,\overline{S}) \stackrel{3}{\rightarrow} (1,0,1,0,S,\overline{S}) \stackrel{4}{\rightarrow} (1,0,1,1,S,\overline{S}) \stackrel{S}{\rightarrow} (1,0,1,1,T,\overline{S})$.
	\item $(1,0,1,1,T,\overline{S}) \stackrel{4}{\rightarrow} (1,0,1,0,T,\overline{S}) \stackrel{1}{\rightarrow} (0,0,1,0,T,\overline{S}) \stackrel{\overline{S}}{\rightarrow} (0,0,1,0,T,\overline{T})$.
\end{enumerate}

It is easy to see that no point can
occur in two different chunks, because the sets on $\cH$ or $\overline{\cH}$ are different.
So, the cycle is simple. The number of chunks is at least the number of
points in the $m$-dimensional hypercube. The matching property should be clear.
\qed
\end{proof}

We now define the function $f$. 
Let the directed path consisting of the first two edges of chunk $K_i$
be $\cP_i$. Note that $\cP_{2i}$ is downward and $\cP_{2i+1}$ is upward.
We describe the function $f$ and state many properties of $\df(f)$. 
It will be convenient to have define the following sequences of $4$ bits.
We set $B_1 = (0,0,1,0)$, $B_2 = (1,0,0,0)$, $C_1 = (1,1,1,0)$,
and $C_2 = (1,0,1,1)$. We use $A$ to denote any one of these.

\begin{itemize}
	\item The function $f$ will be defined on 
all the endpoints of the $\cP_i$'s. 
	\item For $\cP_1$, the small endpoint has value $v$ (the exact choice for this is immaterial),
and the larger endpoint has value $v+1$. For $\cP_{2i+1}$ ($i > 0$),
the small end has value $v$ and the large end has value $v+2$.
For $\cP_{2i}$ ($\forall i$), the large end has value $v+2$
and the small end has value $v$.
	\item Fix any $R$. One and only one point of the form $(B_j,R,\overline{R})$ is present in $\df(f)$.
	Similarly, one and only one of $(C_j,R,\overline{R})$ is present in $\df(f)$.
	We also have $(B_j,R,\overline{R}) \in \df(f)$ iff $(C_j,R,\overline{R}) \in \df(f)$.
	No other point is present in levels $1$ and $3$.
	
	\item Fix any $R$. Suppose $s(R) \supset R$.
	One and only one of $(B_j,s(R),\overline{R})$ is present in $R$. 
	Similarly, one and only one of $(C_j,s(R),\overline{R})$ is present in $R$. 
	We also have $(B_j,s(R),\overline{R}) \in \df(f)$ iff $(C_j,s(R),\overline{R}) \in \df(f)$.
	No other point is present in levels $2$ and $4$.
	
	Suppose $s(R) \subset R$. Then these points are of the form $(A,R,\overline{s(R)})$.
	
	\item Pairs of neighbors in $\df(f)$ are either level $1$-level $2$ pairs, or
	level $3$-level $4$ pairs. They are always of the following form:
	$(A,R,\overline{R}) \rightarrow (A,s(R),\overline{R})$ (if $R \subset s(R)$) or
	$(A,R,\overline{R}) \rightarrow (A,R,\overline{s(R)})$ (if $R \supset s(R)$).
	
	\item For any point of $\df(f)$, there is at most one neighbor present in $\df(f)$.
	Hence, any square of $\cB$ contains at most $2$ points of $\df(f)$.
	
	\item Consider some point $(B_j,R,\overline{R})$ in level $1$. The
	only point in level $3$ at a Hamming distance $2$ from this point is 
	$(C_j,R,\overline{R})$. A similar statement holds for points in level $2$.
\end{itemize}

\begin{claim} \label{clm-f} The function $f$ is not submodular-extendable.
\end{claim}

\begin{proof} By Lemma~\ref{lem-matched}, it suffices to show a path
certificate. As the astute reader might have guessed, all the $\cP_i$'s form
such a set. A matching exists because of the fourth property of the cycle $\cC$.
The value of $P_1$ is $1$. The value of any other $P_{2i+1}$ is $2$. Every
$P_{2i}$ has value $-2$. Since the total number of chunks is even,
the value of this set of paths is $-1$.
\qed
\end{proof}

We will now show that $f|_{S}$ for any $S \subset \df(f)$ is extendable. It will
be easiest to show that by proving that any path certificate for $f$
must essentially be the $P_i$'s. 

\begin{claim} \label{clm-unique} Suppose $f$ contains a set of matched paths $\bP$ with no singletons.
This $\bP$ must be the set of all $\cP_i$'s.
\end{claim}

\begin{proof} Consider a point $X$ in $\bP$ that lies in the lowest level (the number
of $1$s in the representation of the point is minimized). We argue that 
this point only has upward edges incident to it. If there is a downward edge 
$e$ incident to it, then $\bP$ must contain an upward edge $e'$ that is matched
to $e$. Therefore, $e' \prec e$ and the lower end of $e'$ must lie in a lower
level than $S$. This contradicts the choice of $S$. Hence, $X$ only
has upward edges incident to it. This means that it can never be in the
interior of a path, and must be a terminal. Therefore, $X \in \df(f)$.
Similarly, points in $\bP$ that lie in the highest level only
have downward edges incident to them, and are also in $\df(f)$.

The points of $\df(f)$ lie in levels $m+1,m+2,m+3,m+4$, called
the $1,2,3,4$ levels. 
Edges between the $1$ and $2$ levels are called \emph{low edges},
those between the $2$ and $3$ levels are \emph{middle edges},
and those between the $3$ and $4$ levels are \emph{high edges}.
All edges of $\bP$ fall into one of these three sets.
Low edges are always upward and high edges are always downward.
Middle edges are matched to either low or high edges.
Therefore, the number of middle edges is exactly the same
as the total number of low and high edges.
Since $\bP$ contains no singletons, every path must
contain at least one middle edge. The total number
of low and high edges in a path is at most $1$.
This implies that \emph{every} path in $\bP$ has
exactly two edges and has one of the two forms:
an upward low and middle edge, or a downward top
and bottom edge. The former paths go from level $1$
to level $3$ and the latter from level $4$ to level $2$.
We must have at least one path of each type to
get both upward and downward edges. Therefore there
is some level $1$ point of $\df(f)$ in $\bP$.

Consider some point $X = (0,0,1,0,R,\overline{R})$ at level $1$ that is a terminal
in $\bP$. Let path $\cQ \in \bP$ start from here.
Note that this is the endpoint for some $\cP_i$, which
is $(0,0,1,0,R,\overline{R}) \rightarrow$ $(1,0,1,0,R,\overline{R}) \rightarrow$
$(1,1,1,0,R,\overline{R})$.
The certificate $\bP$ has an upward path of length $2$
from $X$. The properties of $\df(f)$ tells us that the other end of $\cQ$ can
only be $(1,1,1,0,R,\overline{R})$. It does not
immediately follow that $\cQ$ is $\cP_i$, since
there are two different paths between these points (the endpoints
differ in coordinates $1$ and $2$).
But observe that the second edge of $\cQ$ must be
matched by an downward edge between levels $4$ and $3$.
This edge has an endpoint in level $4$ that must be 
a neighbor of $(1,1,1,0,R,\overline{R})$. By the
properties of $\df(f)$, this point must
be $(1,1,1,0,s(R),\overline{R})$ (assuming $s(R) \supset R$). 
All downward paths of length $2$ from this point
end at $(1,0,0,0,s(R),\overline{R})$. The path changes in coordinates
$2$ and $3$. Since the second edge of $\cQ$ is matched
to the first edge of this path, both of these edges must
be along coordinate $2$. Hence, $\cQ$ is $\cP_{i}$, and
$\cP_{i+1(\textrm{mod} M)}$ also lies in $\bP$. Repeating the argument,
we get that all $\cP_i$'s lie in $\bP$. This completes
the proof.
\qed
\end{proof}

\begin{proof} (Theorem~\ref{thm-min}) By Claim~\ref{clm-f}, the function
$f$ is not submodular-extendable. For some subset $\cA \subset \df(f)$,
suppose $f|_\cA$ is not submodular-extendable. Since $\df(f)$ contains
no squares, by Claim~\ref{clm-single}, there is a path certificate $\bP$
in $\df(f|_\cA)$ that contains no singletons. Note that $\bP$
is also a path certificate for $f$. By Claim~\ref{clm-unique}, $\bP$
contains all $\cP_i$s. But that means that $\bP$ contains all points
in $\df(f)$. Contradiction.
\qed
\end{proof}

\section{From monotonicity to submodularity}
\label{sec:reduction}

In this section, we show a simple reduction from testing monotonicity to testing submodularity.

\begin{lemma} \label{lem-red} Given $f:\{0,1\}^n \rightarrow \RR$, there
exists a function $g:\{0,1\}^{n+1} \rightarrow \RR$ with the following
properties:
\begin{compactitem}
	\item If $f$ is monotonically non-increasing, then $g$ is submodular.
	\item If $f$ is $\epsilon$-far from being monotonically non-increasing, then $g$ is $\epsilon/2$-far	from being submodular.
	\item The value $g(x)$ can be computed by looking at $2$ values of $f$.
\end{compactitem}
\end{lemma}

\begin{proof} We will use small letters $x,y$, etc. to denote points
in $\{0,1\}^n$. Points in $\{0,1\}^{n+1}$ will be denoted by $(0,x)$ or $(1,x)$,
where the first bit denotes the absence or presence of the new element.
We use $\be_*$ to denote the unit vector corresponding to the new element,
and $\be_i,\be_j$ to denote the other unit vectors.
For convenience, monotone will mean monotonically non-increasing.
Define $h(x) = f(\emptyset)\|x\|_1(n-\|x\|_1)$. We define $g$ by the following:
$g(0,x) = h(x)$, and $g(1,x) = f(x) + h(x)$. So any value of $g$
can be computed by looking at $2$ values of $f$.

We first show that $h$ is submodular. Consider $x$ and $i,j$ such
that $x_i = x_j = 0$. Let $\|x\|_1 = r$ and $f(\emptyset) = M$. 
\begin{eqnarray*}
& & h(x+\be_i) + h(x+\be_j) - h(x+\be_i+\be_j) - h(x) \\
& = & M[2(r+1)(n-r-1) - r(n-r) - (r+2)(n-r-2)] \\
& = & M[(2nr-2r^2-2r+2n-2r-2) - nr+r^2 - nr+r^2+2r - 2n+2r+4] \\
& = & M[(2nr-2r^2+2n-4r-2) - (2nr-2r^2+2n-4r-4)] = 2M
\end{eqnarray*}
Hence $h$ is submodular. 

Assume that $f$ is monotone.
Then, for any $x$, $f(x) \leq f(\emptyset) = M$
Since $f(x+\be_i) + f(x+\be_j) - f(x+\be_i+\be_j) - f(x) \leq 2M$,
$f+h$ is also submodular.

Suppose $g$ is not submodular. Then there exists a violated square in $g$. Suppose this square does
not involve $\be_*$. This square is contained in a copy of $\{0,1\}^n$ where the function
is equal to $h$ or $f+h$. But this would imply that either $h$
or $f+h$ is non-submodular. So, this square must involve $\be_*$.
Then we have the following:
$$	0 < g(0,x) + g(1,x+\be_i) - g(0,x+\be_i) - g(1,x) =
f(x+\be_i) - f(x).  $$
This violates the non-increasing property of $f$. Hence, we conclude
that $g$ is submodular.

Now, suppose that $f$ is $\epsilon$-far from being monotone. 
Furthermore, suppose we can modify $\epsilon 2^n$ values
of $g$ to get a submodular function $g'$. Consider the function
$f'(x) = g'(1,x) - g'(x)$. Since $g'$ is submodular, $f'$ must
be monotone. Since $g'$ differs from $g$ in at most $\epsilon 2^n$
values, the monotone function $f'$ differs from $f$ in at most $\epsilon 2^n$
values. This is a contradiction. So, $g$ must
be $\epsilon/2$-far from being submodular.
\qed
\end{proof}

By the results in~\cite{FLNRR02}, there is an $\Omega(\sqrt{n})$ non-adaptive
and $\Omega(\log {n})$ lower bound for $1$-sided monotonicity testers.
We get the following corollary.

\begin{corollary} \label{cor-lb} Any non-adaptive $1$-sided tester for submodularity requires $\Omega(\sqrt{n})$ queries.
Any adaptive $1$-sided tester requires $\Omega(\log n)$ queries.
\end{corollary}

\section{Future work}
\label{sec:future}

All of this work is centered on the following very general question: what really
makes a function submodular? Of course, it is ``just" monotonicity
of marginal values, but this does not capture the full structure
of submodular functions. We want to understand how different
sets of values in a submodular function interact and influence each other.
The problem of property testing submodularity appears
to be a very appealing way of studying this question.
Our constructions show that functions
far from being submodular could have marginal values that are much closer
to being monotone. 

The problem of completing
partial functions comes up when we try to understand how to convert a non-submodular
function into a submodular one (a major component of a property testing proof).
Again, our constructions yield insight into how seemingly unconnected 
parts of a submodular function must be related.

The authors believe there is a lot of scope for further research directions. There are many interesting
questions to be answered, and we have barely seen the tip of the iceberg.
We state some questions here.

1. \emph{Relation between violated squares and distance to submodularity:} For a function
$f$ $\epsilon$-far from being submodular, what is the minimum (as a function of $\epsilon$ and $n$)
density of violated squares it can have? Can we prove that this minimum density is at least $\poly(\epsilon/n)$?

2. \emph{Efficient testers for submodularity:} Does there exist a tester for submodularity
with running time $\poly(n/\epsilon)$ or maybe $\poly(n)$ for constant $\epsilon$? Perhaps we
can find an efficient \emph{adaptive} tester, or a tester that searches for obstructions
other than violated squares.

3. \emph{Testing rank functions:} A matroid gives rise to a \emph{rank function}, which
is always submodular. A function is a rank function iff it is a submodular function
with marginal values $0$ or $1$. Can we test whether an input function $f$ is a rank function?
Note that even though these are a special case of submodular functions, it is not
clear that this is easier (or harder). This is because the \emph{distance
to a rank function} might be significantly different from the distance to submodularity.

4. \emph{Testing matroid independence oracles:} Any matroid can be represented as a collection
of independent sets. Suppose we have a function that tells us whether a set
is independent (for some purported matroid). Can we efficiently test whether
this function is indeed a valid independence oracle? This seems like a rather fundamental
question about matroids.

\paragraph{Acknowledgement.}

We thank Deeparnab Chakrabarty for very useful discussions. Indeed, the main question
whether submodularity is testable came up during discussions with him.

\bibliographystyle{alpha}
\bibliography{submodular_testing}

\end{document}